\documentclass[a4paper,12pt]{article}
\usepackage{color,amsmath,amssymb,path,amsthm,caption,url,changepage,algorithmic,eqnarray,wasysym,textcomp,enumerate,setspace,graphicx,mathtools}

\usepackage[margin=1in]{geometry}

\def\A{{\cal A}}
\def\B{{\cal B}}
\def\C{{\cal C}}

\def\S{{\cal S}}

\def\bd{{\partial}}
\def\Vor{{\rm Vor}}

\def\reals{\mathbb{R}}

\def\etal{\textsl{et~al.}}

\newtheorem{theorem}{Theorem}[section]
\newtheorem{lemma}[theorem]{Lemma}

\newtheorem{proposition}[theorem]{Proposition}

\begin{document}

\begin{titlepage}

\title{The 2-Center Problem in Three Dimensions{\large{\thanks{%
Work on this paper by Pankaj Agarwal and Micha Sharir has been
supported by Grant 2006/194 from the U.S.-Israeli Binational Science
Foundation. Work by Pankaj Agarwal is also supported by NSF under
grants
 CNS-05-40347, CCF-06 -35000, IIS-07-13498,
                and CCF-09-40671, by ARO grants
                W911NF-07-1-0376 and W911NF-08-1-0452, by an
                NIH grant 1P50-GM-08183-01, and by a DOE grant
                OEG-P200A070505
Work by Micha Sharir has also been supported
by NSF Grants CCF-05-14079 and CCF-08-30272,
by Grants 155/05 and 338/09 from the Israel Science Fund,
and by the Hermann Minkowski--MINERVA Center for Geometry at Tel Aviv
University. A preliminary version of this paper appeared in {\it Proc. 26th Sympos. on Computational Geometry} 2010, pp. 87--96.
  }}}}

\author{
  Pankaj K. Agarwal\thanks{%
          Department of Computer Science, Box 90129, Duke University,
      Durham, NC 27708-0129, USA; {\tt pankaj@cs.duke.edu}}
       \and
       Rinat Ben Avraham\thanks{%
          School of Computer Science, Tel Aviv University, Tel Aviv 69978,
          Israel; {\tt rinatba@gmail.com}}
       \and
       Micha Sharir\thanks{%
          School of Computer Science, Tel Aviv University, Tel~Aviv 69978,
          Israel; and Courant Institute of Mathematical Sciences, New York
          University, New York, NY~~10012,~USA; {\tt michas@post.tau.ac.il }}
 }
\date{}
\maketitle

\begin{abstract}
Let $P$ be a set of $n$ points in $\reals^3$.  The \emph{2-center
problem} for $P$ is to find two congruent balls of minimum radius
whose union covers $P$. We present two randomized algorithms for
computing a 2-center of $P$. The first  algorithm runs in
$O(n^3\log^5 n)$ expected time, and the second algorithm runs in
$O((n^2 \log^5 n) /(1-r^*/r_0)^3)$ expected time, where $r^*$ is the
radius of the 2-center balls of $P$ and $r_0$ is the radius of the
smallest enclosing ball of $P$. The second algorithm is faster than
the first one as long as $r^*$ is not too close to $r_0$, which is
equivalent to the condition that the centers of the two covering
balls be not too close to each other.
\end{abstract}

\end{titlepage}

\section{Introduction}
\label{sec:introduction}

\subsection{Background}
\label{subsec:background} Let $P = \{p_1,\ldots,p_n\}$ be a set of
$n$ points in $\reals^3$. The \emph{2-center problem} for $P$ is to
find two congruent balls of minimum radius whose union covers $P$.
This is a special case of the general $p$-center problem in
$\reals^d$, which calls for covering a set $P$ of $n$ points in
$\mathbb{R}^d$ by $p$ congruent balls of minimum radius. If $p$ is
part of the input, the problem is known to be NP-complete~\cite{MK}
even for $d=2$, so the complexity of algorithms for solving the
$p$-center problem, for any fixed $p$, is expected to increase more
than polynomially in $p$. Agarwal and Procopiuc showed that the
$p$-center problem in $\reals^d$ can be solved in $n^{O(p^{1-1/d})}$
time~\cite{AP}, improving upon a naive $n^{O(p)}$-solution. At the
other extreme end, the 1-center problem (also known as the
\emph{smallest enclosing ball} problem) is known to be an LP-Type
problem, and can thus be solved in $O(n)$ randomized expected time
in any fixed dimension, and also in deterministic linear time
~\cite{CM, NML, NMLT}. Faster approximate solutions to the general
$p$-center problem have also been proposed~\cite{AP, BHI, BE}.

If $d$ is not fixed, the 2-center problem in $\reals^d$ is
NP-Complete~\cite{MK2}. The 2-center problem in $\reals^2$ has a
relatively rich history, mostly in the past two decades. Hershberger
and Suri~\cite{HS} showed that the decision problem of determining
whether $P$ can be covered by two disks of a given radius $r$ can be
solved in $O(n^2 \log n)$ time. This has led to several
nearly-quadratic algorithms~\cite{ASP, DE, JK} that solve the
optimization problem, the best of which, due to Jaromczyk and
Kowaluk~\cite{JK}, runs in $O(n^2 \log n)$ deterministic time.
Sharir~\cite{MS} considerably improved these bounds and obtained a
deterministic algorithm with $O(n \log^9 n)$ running time. His
algorithm combines several geometric techniques, including
parametric searching, searching in monotone matrices, and dynamic
maintenance of planar configurations. Chan~\cite{TC} (following an
improvement by Eppstein~\cite{DEF}) improved the running time to
$O(n \log^2 n \log^2 \log n)$.

The only earlier work on the 2-center problem in $\reals^3$ we are aware of is by Agarwal~\etal~\cite{AES},
which presents an algorithm with $O(n^{3+\varepsilon})$ running time, for any $\varepsilon > 0$. It uses
a rather complicated data structure for dynamically maintaining upper and lower envelopes of bivariate functions.

\subsection{Our results}
\label{subsec:results} We present two randomized algorithms for the
2-center problem in $\reals^3$. We first present an algorithm whose
expected running time is $O(n^3 \log^5 n)$. It is conceptually a
natural generalization of the earlier algorithms for the planar
2-center problem~\cite{ASP, DE, JK}; its implementation however is
considerably more involved. The second algorithm runs in $O((n^2
\log^5 n) /(1-r^*/r_0)^3)$ expected time, where $r^*$ is the common
radius of the 2-center balls and $r_0$ is the radius of the smallest
enclosing ball of $P$.  This is based on some of the ideas in
Sharir's planar algorithm~\cite{MS}, but requires several new
techniques. As in the previous algorithms, we first present
algorithms for the decision problem: given $r > 0$, determine
whether $P$ can be covered by two balls of radius $r$. We then
combine it with an adaptation of Chan's randomized optimization
technique~\cite{TCG} to obtain a solution for the optimization
problem. In both cases, the asymptotic expected running time of the
optimization algorithm is the same as that of the decision procedure
(which itself is deterministic).

The paper is organized as follows. Section~\ref{sec:sketches}
briefly sketches our two solutions. Section~\ref{sec:cubic_alg}
presents the near-cubic algorithm, and
Section~\ref{sec:improved_alg} presents the improved algorithm. A
key ingredient of both algorithms is a dynamic procedure for testing
whether the intersection of a collection of balls in $\mathbb{R}^3$
is nonempty. We present the somewhat technical details of this
procedure in Section~\ref{sec:spherical_polytopes}, and conclude in
Section~\ref{sec:discussion} with a few open problems.

\section{Sketches of the Solutions}
\label{sec:sketches}

\subsection{The near-cubic algorithm}
\label{subsec:n^3_sketch} To solve the decision problem, in the less
efficient but conceptually simpler manner, we use a standard
point-plane duality, and replace each point $p \in P$ by a dual
plane $p^*$, and each plane $h$ by a dual point $h^*$, such that the
above-below relations between points and planes are preserved. We
note that if $P$ can be covered by two balls $B_1, B_2$ (not
necessarily congruent), then there exists a plane $h$ (containing
the circle $\bd{B_1} \cap \bd{B_2}$, if they intersect at all, or
separating $B_1$ and $B_2$ otherwise) separating $P$ into two
subsets $P_1, P_2$, such that $P_1 \subset B_1$ and $P_2 \subset
B_2$. We therefore construct the arrangement $\A$ of the set $\{p^*
\mid p \in P\}$ of dual planes. It has $O(n^3)$ cells, and each cell
$\tau$ has the property that, for any point $w \in \tau$, its primal
plane $w^*$ separates $P$ into two subsets of points, $P_\tau^+$ and
$P_\tau^-$, which are the same for every $w \in \tau$, and depend
only on $\tau$. We thus perform a traversal of $\A$, which proceeds
from each visited cell to a neighbor cell. When we visit a cell
$\tau$, we check whether the subsets $P_\tau^+$ and $P_\tau^-$ can
be covered by two balls of radius $r$, respectively. To do so, we
maintain dynamically the intersection of the sets $\{B_r(p)\mid p
\in P_\tau^+\}$, $\{B_r(p)\mid p \in P_\tau^-\}$, where $B_r(p)$ is
the ball of radius $r$ centered at $p$, and observe that (a) any
point in the first (resp., second) intersection can serve as the
center of a ball of radius $r$ which contains $P_\tau^+$ (resp.,
$P_\tau^-$), and (b) no ball of radius $r$ can cover $P_\tau^+$
(resp., $P_\tau^-$) if the corresponding intersection is empty.
Moreover, when we cross from a cell $\tau$ to a neighbor cell
$\tau'$, $P_\tau^+$ changes by the insertion or deletion of a single
point, and $P_\tau^-$ undergoes the opposite change, so each of the
sets of balls $\{B_r(p)\mid p \in P_\tau^+\}$, $\{B_r(p)\mid p \in
P_\tau^-\}$ changes by the deletion or insertion of a single ball.
As we know the sequence of updates in advance, maintaining
dynamically the intersection of either of these sets of balls can be
done in an offline manner. Still, the actual implementation is
fairly complicated. It is performed using a variant of the
multi-dimensional parametric searching technique of
Matou\v{s}ek~\cite{JM} (see also~\cite{TCA, CMS, NPT}). The same
procedure is also used by the second improved algorithm. For the
sake of readability, we describe this procedure towards the end of
the paper, in Section~\ref{sec:spherical_polytopes}.

The main algorithm uses a segment tree to represent the sets
$P_\tau^+$ (and another segment tree for the sets $P_\tau^-$).
Roughly, viewing the traversal of $\A$ as a sequence $\Sigma$ of
cells, each ball $B_r(p)$ has a \emph{life-span} (in $P_\tau^+$),
which is a union of contiguous maximal subsequences of cells $\tau$,
in which $p \in P_\tau^+$, and a complementary life-span in
$P_\tau^-$. We store these (connected portions of the) life-spans as
segments in the segment tree. Each leaf of the tree represents a
cell $\tau$ of $\A$, and the balls stored at the nodes on the path
to the leaf from the root are exactly those whose centers belong to
the set $P_\tau^+$ (or $P_\tau^-$). By precomputing the intersection
of the balls stored at each node of the tree, we can express each of
the intersections $\bigcap\{B_r(p) \mid p \in P_\tau^+\}$ and
$\bigcap\{B_r(p) \mid p \in P_\tau^-\}$, for each cell $\tau$, as
the intersection of a logarithmic number of precomputed
intersections (see also~\cite{DE}). We show that such an
intersection can be tested for emptiness in $O(\log^5 n)$ time. This
in turn allows us to execute the decision procedure with a total
cost of $O(n^3 \log^5 n)$. We then return to the original
optimization problem and apply a variant of Chan's randomization
technique~\cite{TCG} to solve the optimization problem by a small
number of calls to the decision problem, obtaining an overall
algorithm with $O(n^3 \log^5 n)$ \emph{expected} running
time.\footnote{\small The earlier algorithm in~\cite{AES} follows
the same general approach, but uses an even more complicated, and
slightly less efficient machinery for dynamic emptiness testing of
the intersection of congruent balls.}

\subsection{The improved solution}
\label{subsec:n^2_sketch} The above algorithm runs in nearly cubic
time because it has to traverse the entire arrangement $\A$, whose
complexity is $O(n^3)$. In Section~\ref{sec:improved_alg} we improve
this bound by traversing only portions of $\A$, adapting some of the
ideas in Sharir's improved solution for the planar
problem~\cite{MS}. Specifically, Sharir's algorithm solves the
decision problem (for a given radius $r$) in three steps, treating
separately three subcases, in which the centers $c_1, c_2$ of the
two covering balls are, respectively, far apart ($|c_1 c_2| > 3r$),
at medium distance apart ($r < |c_1 c_2| \leq 3r$) and near each
other ($|c_1 c_2| \leq r$). We base our solution on the techniques
used in the first two cases, which, for simplicity, we merge into a
single case (as done in~\cite{DEF} for the planar case), and extend
it so that we only need to assume that $|c_1c_2| \geq \beta r$, for
any fixed $\beta > 0$. In more detail, letting $B_r(p)$ denote the
disk of radius $r$ centered at a point $p$, Sharir's algorithm
guesses a constant number of lines $l$, one of which separates the
centers $c_1, c_2$ of the respective solution disks $D_1,D_2$, so
that the set $P_L$ of the points to the left of $l$ is contained in
$D_1$. We then compute the intersection $K(P_L) = \bigcap_{p \in
P_L} B_r(p)$, and intersect each $\bd{B_r(p)}$, for $p \in P_R = P
\setminus P_L$ (the subset of points to the right of $l$), with
$\bd{K(P_L)}$. It is easily seen that $\bd{K(P_L)}$ has linear
complexity and that each circle $\bd{B_r(p)}$, for $p \in P_R$,
intersects it at two points (at most). This produces $O(n)$ critical
points (vertices and intersection points) on $\bd{K(P_L)}$ and
$O(n)$ arcs in between. As argued in~\cite{MS}, it suffices to
search these points and arcs for possible locations of the center of
$D_1$ (and dynamically test whether the balls centered at the
uncovered points have nonempty intersection).

Generalizing this approach to $\reals^3$, we need to guess a separating plane $\lambda$, to retrieve the subset $P_L \subseteq P$ of points to the left of $\lambda$, to compute $\bd{K(P_L)}$ (which, fortunately, still has only linear complexity), to intersect $\bd{B_r(p)}$, for each $p \in P_R$, with $\bd{K(P_L)}$, and to form the arrangement of the resulting intersection curves. Each cell of this arrangement is a candidate for the location of the center of the left covering ball $B_1$, and for each placement in $\tau$, $B_1$ contains the same fixed subset of $P$ (which depends only on $\tau$).

However, the complexity of the resulting arrangement $M_K$ on
$\bd{K(P_L)}$ might potentially be cubic. We therefore compute only
a portion $M$ of $M_K$, which suffices for our purposes, and prove
that its complexity is only $O(n^2)$. This is the main geometric
insight in the improved algorithm, and is highlighted in
Lemma~\ref{lemma:quadratic_K_P_L}. We show that if there is a
solution then $O(1/\beta^3)$ guesses suffice to find a separating
plane. This implies that the running time of the improved decision
procedure is $O((1/\beta^3) n^2\log^5 n)$. Thus, it is nearly
quadratic for any fixed value of $\beta$. We show that one can take
$\beta = 2(r_0/r -1)$, where $r_0$ is the radius of the smallest
enclosing ball of $P$.

To solve the optimization problem, we conduct a search on the
optimal radius $r^*$, using our decision procedure, starting from
small values of $r$ and going up, halving the gap between $r$ and
$r_0$ at each step\footnote{\small We have to act in this manner to
make sure that we do not call the decision procedure with values of
$r$ which are too close to $r_0$, thereby losing control over the
running time.}, until the first time we reach a value $r > r^*$.
Then we use a variant of Chan's technique~\cite{TCG}, combined with
our decision procedure, to find the exact value of $r^*$. The way
the search is conducted guarantees that its cost does not exceed the
bound $O((1/\beta^3) n^2 \log^5 n)$, for the separation parameter
$\beta = 2(r_0/r^* -1)$ for $r^*$. Hence, we obtain a randomized
algorithm that solves the 2-center problem for any positive
separation of $c_1$ and $c_2$, and runs in $O((n^2 \log^5 n)
/(1-r^*/r_0)^3)$ expected time.

\section{A Nearly Cubic Algorithm}
\label{sec:cubic_alg}

\subsection{The decision procedure}
\label{subsec:decision_procedure} In this section we give details of
the implementation of our less efficient solution, some of which are
also applicable for the improved solution. Recall from the
description in Section~\ref{sec:sketches} that the decision
procedure, on a given radius $r$, constructs two segment trees $T^+,
T^-$, on the life-spans of the balls $B_r(p)$, for $p \in P$ (with
respect to the tour of the dual plane arrangement $\A$). Each leaf
is a cell $\tau$ of $\A$, and the balls, whose centers belong to
$P_\tau^+$ (resp., $P_\tau^-$), are those stored at nodes on the
path from the root to $\tau$ in $T^+$ (resp., $T^-$).

For each node $u$ of $T^+$, let $S_u$ denote the intersection of all
the balls (of radius $r$) stored at $u$. We refer to each $S_u$ as a
\emph{spherical polytope}; see~\cite{BCT, BLNP, BN} for (unrelated)
studies of spherical polytopes. We compute each $S_u$ in $O(|S_u|
\log |S_u|)$ deterministic time, using the algorithm by
Br\"{o}nnimann et al.~\cite{BCM} (see also~\cite{CS, ER} for
alternative algorithms). Since the arrangement $\A$ consists of
$O(n^3)$ cells, standard properties of segment trees imply that the
two trees require $O(n^3 \log n)$ storage and $O(n^3 \log^2 n)$
preproccessing time.

Clearly, the intersection $K(P_\tau^+)$ (resp., $K(P_\tau^-)$) of
the balls whose centers belong to $P_\tau^+$ (resp., $P_\tau^-$) is
the intersection of all the spherical polytopes $S_u$, over the
nodes $u$ on the path from the root to $\tau$ in $T^+$ (resp.,
$T^-$).

\paragraph{Intersection of spherical polytopes.}
%
Let $\S = \{S_1,\ldots,S_t\}$ be the set of $t = O(\log n)$
spherical polytopes stored at the nodes of a path from the root to a
leaf of $T^+$ or of $T^-$, where, as above, a spherical polytope is
the intersection of a finite set of balls, all having the common
radius $r$. Each $S_i$ is the intersection of some $n_i$ balls, and
$\sum_{i=1}^t n_i \leq n$. Our current goal is to determine, in
polylogarithmic time, whether the intersection $K$ of the spherical
polytopes in $\S$ is nonempty. If this is the case for at least one
path of $T^+$ and for the same path in $T^-$ then $r^* \leq r$, and
otherwise $r^* >r$. Moreover, if there exist a pair of such paths
for which both intersections have nonempty interior, then $r^* < r$
(because we can then slightly shrink the balls and still get a
nonempty intersection). If no such pair of paths have this property,
but there exist pairs with nonempty intersections (with at least one
of them being degenerate) then $r^* = r$.

The algorithm for testing emptiness of $K$ is technical and fairly
involved. For the sake of readability, we delegate its description
to Section~\ref{sec:spherical_polytopes}. It uses a variant of
multidimensional parametric searching which somewhat resembles
similar techniques used in ~\cite{TCA, CMS, JM, NPT}. It is
essentially independent of the rest of the algorithm (with some
exceptions, noted later). We summarize it in the following
proposition.

\begin{proposition}
\label{prop:intersection_test} Let $\S$ be a collection of spherical
polytopes, each defined as the intersection of at most $n$ balls of
a fixed radius $r$. Let $N$ denote the sum, over the polytopes of
$\S$, of the number of balls defining each polytope. After a
preprocessing stage, which takes $O(N \log n)$ time and uses $O(N)$
storage, we can test whether any $t \leq \log n$ polytopes of $\S$
have a nonempty intersection in $O(\log^5 n)$ time, and also
determine whether the intersection has nonempty interior.
\end{proposition}

Hence, we check, for each cell $\tau$, whether each of $K(P_\tau^+)$
and $K(P_\tau^-)$ are nonempty and non-degenerate. To this end, we
go over each path of $T^+$, and over the same path of $T^-$, and
check, using the procedure described in
Proposition~\ref{prop:intersection_test}, whether the spherical
polytopes along the tested paths (of $T^+$ and of $T^-$) have a
nonempty intersection (and whether these intersections have nonempty
interiors). We stop when a solution for which both $K(P_\tau^+)$ and
$K(P_\tau^-)$ are nonempty and non-degenerate is obtained, and
report that $r^* < r$. Otherwise, we continue to test all cells
$\tau$. If at least one degenerate solution is found (i.e., a
solution where both $K(P_\tau^+), K(P_\tau^-)$ are nonempty, and at
least one of them has nonempty interior), we report that $r^* = r$,
and otherwise $r^* > r$.

By proposition~\ref{prop:intersection_test}, the cost of this
procedure is $O(n^3 \log^5 n)$. This subsumes the cost of all the
other steps, such as constructing the arrangement $\A$ and the
segment trees $T^+, T^-$. We therefore get a decision procedure
which runs in $O(n^3 \log^5 n)$ (deterministic) time.

\subsection{Solving the optimization problem}
\label{subsec:optimization_problem}
We now combine our decision procedure with the randomized
optimization technique of Chan~\cite{TCG}, to obtain an algorithm
for the optimization problem, which runs in $O(n^3 \log^5 n)$
\emph{expected} time. Our application of Chan's technique, described
next, is somewhat non-standard, because each recursive step has also
to handle global data, which it inherits from its ancestors.

Chan's technique, in its ``purely recursive'' form, takes an optimization problem that has to compute an optimum value $w(P)$ on an input set $P$. The technique replaces $P$ by several subsets $P_1,\ldots,P_s$, such that $w(P) = \min\{w(P_1),\ldots,w(P_s)\}$, and $|P_i| \leq \alpha|P|$ for each $i$ (here $\alpha < 1$ and $s$ are constants). It then processes the subproblems $P_i$ in a \emph{random} order, and computes $\displaystyle\min_i w(P_i)$ by comparing each $w(P_i)$ to the minimum $w$ collected so far, and by replacing $w$ by $w(P_i)$ if the latter is smaller.\footnote{\small So the value of $w$ keeps shrinking.} Comparisons are performed by the decision procedure, and updates of $w$ are computed recursively. The crux of this technique is that the expected number of recursive calls (in a single recursive step) is only $O(\log s)$, and this (combined with some additional enhancements, which we omit here) suffices to make the expected cost of the whole procedure asymptotically the same as the cost of the decision procedure, for \emph{any} values of $s$ and $\alpha$. Technically, if the cost $D(n)$ of the decision procedure is
$\Omega(n^\gamma)$, where $\gamma$ is some fixed positive constant, the expected running time is $O(D(n))$ provided that
\begin{equation}
\label{equation:lnr}
(\ln s + 1) \alpha^\gamma < 1.
\end{equation}
However, even when (\ref{equation:lnr}) does not hold ``as
is'', Chan's technique enforces it by compressing $l$ levels of the
recursion into a single level, for $l$ sufficiently large, so its expected cost is still $O(D(n))$. See~\cite{TCG} for details.

To apply Chan's technique to our decision procedure,
we pass to the dual space, where each point $p \in P$ is mapped to a
plane $p^*$, as done in the decision procedure. We obtain the set $P^* =
\{p^* \mid p \in P\}$ of dual planes, and we consider its arrangement
$\A = \A(P^*)$, where each cell $\tau$ in $\A$ represents an
equivalence class of planes in the original space, which separate $P$
into the same two subsets of points $P_\tau^+, P_\tau^-$.

To decompose the optimization problem into subproblems, as required
by Chan's technique, we construct a $(1/\varrho)$-\emph{cutting} of
the dual space. We recall that, given a collection $\it{H}$ of $n$
hyperplanes in $\mathbb{R}^d$ and a parameter $1 \leq \varrho \leq
n$, a $(1/\varrho)$-\emph{cutting} of $\A(\it{H})$ of size $q$ is a
partition of space into $q$ (possibly unbounded) openly disjoint
$d$-dimensional simplices $\Delta_1,\ldots,\Delta_q$, such that the
interior of each simplex $\Delta_i$ is intersected by at most
$n/\varrho$ of the hyperplanes of $\it{H}$. See~\cite{JMC} for more
details.  We use the following well known result~\cite{BC, CF}:

\begin{lemma}
\label{lemma:cutting}
Given a set $\it{H}$ of $n$ hyperplanes in $\mathbb{R}^d$, a
$(1/\varrho)$-cutting of $\A(\it{H})$ of size $O(\varrho^d)$ can be constructed
in time $O(n\varrho^{d-1})$, for any $\varrho \leq n$.
\end{lemma}

Returning to our setup, we construct a $(1/\varrho)$-cutting for
$\A(P^*)$, for a specific constant value of $\varrho$, that we will
fix later, and obtain $O(\varrho^3)$ simplices, such that the
interior of each of them is intersected by at most $n/\varrho$
planes of $P^*$. Each simplex $\Delta_i$ corresponds to one
subproblem and contains some (possibly only portions of) cells
$\tau_1,\ldots,\tau_k$ of the arrangement $\A$. We recall that each
cell $\tau_j$ represents an equivalence class of planes which
separate $P$ into two subsets of points $P_{\tau_j}^+$ and
$P_{\tau_j}^-$. Hence, $\Delta_i$ represents a collection of such
equivalence classes. All these subproblems have in common the sets
$(P^*)_{\Delta_i}^+$, $(P^*)_{\Delta_i}^-$, consisting,
respectively, of all the planes that pass fully above $\Delta_i$ and
those that pass fully below $\Delta_i$. (These sets are dual to
respective subsets $P_{\Delta_i}^+, P_{\Delta_i}^-$ of $P$, where
$P_{\Delta_i}^+$ is contained in all the sets $P_{\tau_j}^+$, for
the cells $\tau_j$, that meet $\Delta_i$, and symmetrically for
$P_{\Delta_i}^-$.) Note that most of the dual planes belong to
$(P^*)_{\Delta_i}^+ \cup (P^*)_{\Delta_i}^-$; the ``undecided''
planes are those that cross the interior of $\Delta_i$, and their
number is at most $n/\varrho$. We denote the set of these planes as
$(P^*)_{\Delta_i}^0$ (and the set of their primal points as
$P_{\Delta_i}^0$).

To apply Chan's technique, we construct two segment trees on the
arrangement of $(P^*)_{\Delta_i}^0$, as described in Section~\ref{subsec:decision_procedure}. Consider one of these segment trees, $T^+$,
that maintains the set of balls $\B^+ = \{B_r(p) \mid p \in
P_{\tau_j}^+\}$. Each cell $\tau_j$ in $\Delta_i$ is represented by a
leaf of $T^+$. Each ball is represented as a collection of disjoint
life-spans, with respect to a fixed tour of the cells of $\A((P^*)_{\Delta_i}^0)$,
which are stored as segments in $T^+$, as described earlier.
In addition, we compute the intersection of the balls centered
at the points of $P_{\Delta_i}^+$, in $O(n \log n)$ time, and store it at
the root of $T^+$. Note that, as we go down the recursion, we keep
adding planes to $(P^*)_{\Delta_i}^+$, that is, points to $P_{\Delta_i}^+$, and the
actual set $P_{\Delta_i}^+$ of points dual to the planes above the
current $\Delta_i$ is the union of logarithmically many subsets, each
obtained at one of the ancestor levels of the recursion, including the current step.
However, we cannot inherit the precomputed intersections of the balls
in these subsets of $P_{\Delta_i}^+$ from the previous levels, since, as we go down the
recursion, Chan's technique keeps `shrinking' the radius of the balls. Hence, each
time we have to solve a decision subproblem, we compute the
intersection of the balls centered at the points of $P_{\Delta_i}^+$
(collected over all the higher levels of the recursion)
from scratch. (See below for details on the additional cost incurred
by this step.) We build a second segment tree $T^-$
that maintains the balls of $\B^- = \{B_r(p) \mid p \in
P_{\tau_j}^-\}$, in a fully analogous manner. The running time so far (of the decision procedure)
is $O(n \log n + m^3 \log^2 m)$, where $m$ is the
number of planes in $(P^*)_{\Delta_i}^0$ and $n$ is the size of the
initial input set $P$.

To solve the decision procedure for a given subproblem associated
with a simplex $\Delta_i$, we test, by going over all the
root-to-leaf paths in $T^+$ and $T^-$, whether there exists a cell
$\tau$ (overlapping $\Delta_i$), for which the intersections of the
spherical polytopes on the two respective paths in $T^+$ and $T^-$
are nonempty (and, if nonempty, whether they both have nonempty
interiors). The overall cost of this step, iterating over the
$O(m^3)$ cells of $\A((P^*)_{\Delta_i}^0)$ and applying the
procedure from Section~\ref{subsec:decision_procedure} for
intersecting spherical polytopes, is $O(m^3 \log^5 n)$.

When the recursion bottoms out, we have two subsets
$P_{\Delta_i}^+$, and $P_{\Delta_i}^-$ of $O(n)$ points, and a
constant number of points in $P_{\Delta_i}^0$. Hence, we try the
constant number of possible separations of $P_{\Delta_i}^0$ into an
ordered pair of subsets $P_1$ and $P_2$, and, for each of these separations, we compute the two
smallest enclosing balls of the sets $P_{\Delta_i}^+ \cup P_1$ and
$P_{\Delta_i}^- \cup P_2$ in linear time. If both $P_{\Delta_i}^+
\cup P_1$ and $P_{\Delta_i}^- \cup P_2$ can be covered by balls of
radius $r$, for at least one of the possible separations of
$P_{\Delta_i}^0$ into two subsets, then we have found a solution for
the 2-center problem. (Discriminating between $r^* = r$ or $r^* < r$
is done as in Section~\ref{subsec:decision_procedure}.)

We now apply Chan's technique to this decision procedure. Note that
this application is not standard because the recursive subproblems are
not ``pure'', as they also involve the ``global'' parameter $n$. We
therefore need to exercise some care in the analysis of the expected
performance of the technique.

Specifically, denote by $T(m,n)$ an upper bound on the expected
running time of the algorithm, for preprocessing
a recursive subproblem involving $m$ points, where the initial input
consists of $n$ points. Then $T(m,n)$ satisfies the following recurrence.

\begin{equation}
\label{eqn:T(m,n)}
T(m,n) \leq \left\{ \begin{array}{ll}
\ln(c\varrho^3)T(m/\varrho,n) + O(m^3 \log^5 n + n\log n),   & \mbox{for $m \geq \varrho$,}\\
O(n),                                        & \mbox{for $m <
\varrho$,}
\end{array}\right.
\end{equation}
where $c$ is an appropriate absolute constant (so that $c\varrho^3$
bounds the number of cells of the cutting), and $\varrho$ is chosen
to be a sufficiently large constant so that (\ref{equation:lnr})
holds (with $s = c\varrho^3, \alpha = 1/\varrho$, and $\gamma = 3$).
It is fairly routine (and we omit the details) to show that the
recurrence (\ref{eqn:T(m,n)}) yields the overall bound $O(n^3 \log^5
n)$ on the expected cost of the initial problem; i.e., $T(n,n) =
O(n^3 \log^5 n)$. We thus obtain the following intermediate result.

\begin{theorem}
Let $P$ be a set of $n$ points in $\mathbb{R}^3$. A 2-center for $P$
can be computed in $O(n^3 \log^5 n)$ randomized expected time.
\end{theorem}

\section{An Improved Algorithm}
\label{sec:improved_alg}

\subsection{An improved decision procedure $\Gamma$}
\label{subsec:improved_decision_procedure}

\begin{figure}[htbp]
\begin{center}
\input{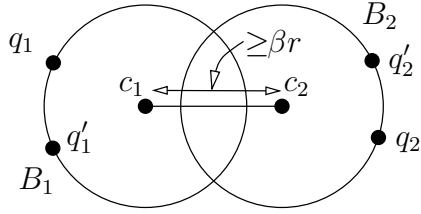}
\caption{\small \sf The points $q_1, q_1', q_2, q_2'$ prevent $|c_1
c_2|$ from getting smaller.} \label{figure:beta}
\end{center}
\end{figure}

Consider the decision problem, where we are given a radius $r$ and a parameter $\beta > 0$, and
have to determine whether $P$ can be covered by two balls of radius
$r$, such that the distance between their centers $c_1, c_2$ is at
least $\beta r$. (Details about supplying a good lower bound for $\beta$ will be given in Section~\ref{sec:separated_centers_optimization}.)
By this we mean that there is no placement of two balls of radius $r$, which cover $P$, such that the distance between their centers is smaller than $\beta r$; see Figure~\ref{figure:beta}.

This assumption is easily seen to imply the following property: Let
$C_{12}$ denote the intersection circle of $\bd{B_1}$ and $\bd{B_2}$
(assuming that $B_1 \cap B_2 \neq \emptyset$). Then any hemisphere
$\nu$ of $\bd{B_1}$, such that (a) the plane $\pi$ through $c_1$
delimiting $\nu$ is disjoint from $C_{12}$, and (b) $\nu$ and
$C_{12}$ lie on different sides of $\pi$, must contain a point $q$
of $P$, for otherwise we could have brought $B_1$ and $B_2$ closer
together by moving $c_1$ in the normal direction of $\pi$, into the
halfspace containing $c_2$ (and $C_{12}$). See
Figure~\ref{figure:case2_assumption}.

\begin{figure}[htbp]
\begin{center}
\input{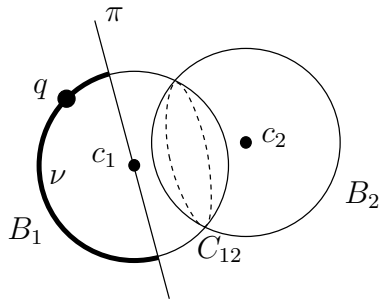}
\caption{\small \sf The plane $\pi$ passes through $c_1$ and is
disjoint from $C_{12}$. The hemisphere $\nu$ delimited by $\pi$,
which lies on the side of $\pi$ not containing $C_{12}$, must
contain a point $q$ of $P$.} \label{figure:case2_assumption}
\end{center}
\end{figure}

\begin{figure}[htbp]
\begin{center}
\input{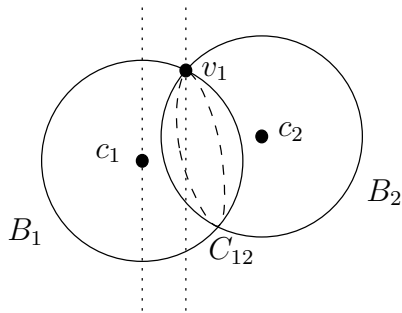}
\caption{\small \sf $v_1$ is the leftmost point of the intersection
circle $C_{12}$.} \label{figure:new_case2_1}
\end{center}
\end{figure}

\smallskip
\noindent{\bf Guessing orientations and separating planes.} We
choose a set $D$ of canonical orientations, so that the maximum
angular deviation of any direction $u$ from its closest direction in
$D$ is an appropriate multiple $\alpha$ of $\beta$. The connection
between $\alpha$ and $\beta$ is given by the following reasoning.
Fix a direction $v \in D$ so that the angle between the orientation
of $c_1 c_2$ and $v$ is at most $\alpha$. Rotate the coordinate
frame so that $v$ becomes the $x$-axis. As above, let $C_{12}$
denote the intersection circle of $\bd{B_1}$ and $\bd{B_2}$
(assuming that the balls intersect). Let $v_1$ be the leftmost point
of $C_{12}$ (in the $x$-direction); see
Figure~\ref{figure:new_case2_1}. If $B_1$ and $B_2$ are disjoint
(which only happens when $|c_1 c_2| > 2r$) we define $v_1$ to be the
leftmost point of $B_2$. To determine the value of $\alpha$, we note
that (in complete analogy with Sharir's algorithm in the
plane~\cite{MS}) our procedure will try to find a $yz$-parallel
plane, which separates $c_1$ from $v_1$. For this, we want to ensure
that $x(v_1) - x(c_1) > \beta r/4$, say, to leave enough room for
guessing such a separating plane. Let $\theta$ denote the angle
$\varangle v_1c_1c_2$ (see Figure~\ref{figure:new_case2_2}). Using
the triangle inequality on angles, the angle between
$\overrightarrow{c_1 v_1}$ and the $x$-axis is at most $\theta +
\alpha$, so $x(v_1) - x(c_1) \geq r \cos(\theta + \alpha)$. Hence,
to ensure the above separation, we need to choose $\alpha$, such
that $\cos(\theta + \alpha) > \beta/4$. Since $|c_1 c_2| \geq \beta
r$, we have $\cos \theta \geq \beta/2$. Hence, it suffices to choose
$\alpha$, such that
$$\alpha \leq \cos^{-1} \frac{\beta}{4} - \cos^{-1} \frac{\beta}{2} =
    \sin^{-1} \frac{\beta}{2} - \sin^{-1} \frac{\beta}{4}
    = \Theta(\beta).$$
With this constraint on $\alpha$, the size of $D$ is
$\Theta\left(1/\alpha^2\right) = \Theta\left(1/\beta^2\right)$.

\begin{figure}[htbp]
\begin{center}
\input{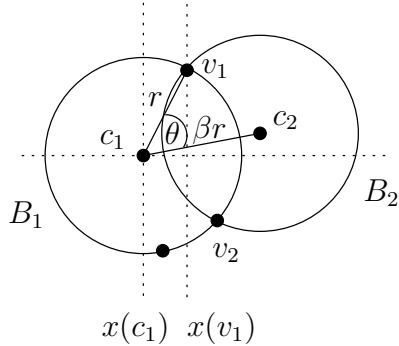}
\caption{\small \sf $x(v_1) - x(c_1) \geq r \cos(\theta + \alpha)$.}
\label{figure:new_case2_2}
\end{center}
\end{figure}

We draw $O(1/\beta)$ $yz$-parallel planes, with horizontal
separation of $\beta r/4$, starting at the leftmost point of $P$
(with respect to the guessed orientation). One of these planes will
separate $v_1$ from $c_1$. Thus, the total number of guesses that we
make (an orientation in $D$ and a separating plane) is
$O(1/\beta^3)$. The following description pertains to a correct
guess, in which the properties that we require are satisfied. (If
all guesses fail, the decision procedure has a negative answer.)

\smallskip
\noindent{\bf Reducing to a 2-dimensional search.}
By the property noted above, the left hemisphere $\nu_{\lambda_0}$ of $\bd{B_1}$, delimited by
the $yz$-parallel plane $\lambda_0$ through $c_1$,
must pass through at least one point $q$ of $P$ (see Figure~\ref{figure:new_case2_left}).

\begin{figure}[htbp]
\begin{center}
\input{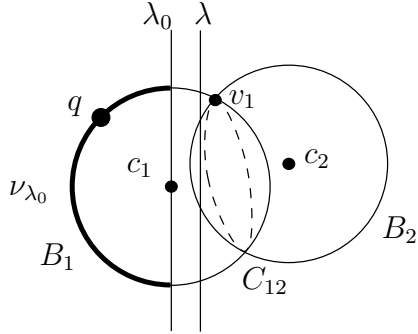}
\caption{\small \sf The separating plane $\lambda$ and its parallel
copy $\lambda_0$ through $c_1$. The hemisphere $\nu_{\lambda_0}$ of
$\bd{B_1}$ to the left of $\lambda_0$ must contain a point $q$ of
$P$.} \label{figure:new_case2_left}
\end{center}
\end{figure}

Let $P_L$ denote the subset
of points of $P$ lying to the left of $\lambda$. Then $P_L$
must be fully
contained in $B_1$ and contain $q$.
We compute the intersection $K(P_L) = \bigcap \{B_r(p) \mid p \in P_L\}$
in $O(n \log n)$ time~\cite{BCM}. If $K(P_L)$ is empty, then $P_L$ cannot be covered
by a ball of radius $r$ and we determine that the currently assumed
configuration does not yield a positive solution for the decision problem.
Otherwise, since $P_L \subseteq B_1$, $c_1$
must lie in
$K(P_L)$. Moreover, since $q \in P_L$ lies on the \emph{left} portion
of $\bd{B_1}$, $c_1$ must
lie on the \emph{right} portion of the boundary of $K(P_L)$. Finally, since $c_1$ lies to the
left of $\lambda$, only the portion $\sigma_L$ of the right part of $\bd{K(P_L)}$ to the left of $\lambda$ has to be considered.
If $K(P_L)$ is disjoint from $\lambda$ then $\sigma_L$ is just the right
portion of $\bd{K(P_L)}$. Otherwise, $\sigma_L$ has a ``hole'', bounded by $\bd{K(P_L)} \cap \lambda$, which is a convex piecewise-circular curve, being the boundary of the intersection of the disks $B_r(p) \cap \lambda$, for $p \in P_L$.

We partition $\sigma_L$ into \emph{quadratically many} cells, such
that if we place the center $c_1$ of the left solution ball $B_1$ in
a cell $\tau$, then, no matter where we place it within $\tau$,
$B_1$ will cover the same subset of points from $P$. To construct
this partition, we intersect, for each $p \in P_R = P \setminus
P_L$, the sphere $\bd{B_r(p)}$ with $\sigma_L$ and obtain a curve
$\gamma_p$ on $\sigma_L$; this curve bounds the portion of the
unique face of $\bd{K(P_L \cup \{p\})}$ within $\sigma_L$. Hence,
within $K(P_L)$, it is a closed connected curve (it may be
disconnected within $\sigma_L$, though). Let $M$ denote the
arrangement formed on $\sigma_L$ by the curves $\gamma_p$, for $p
\in P_R$, and by the arcs of $\sigma_L$. Apriori, $M$ might have
cubic complexity, if many of the $O(n^2)$ pairs of curves $\gamma_a,
\gamma_b$, for $a,b \in P_R$, traverse a linear number of common
faces of $\sigma_L$, and intersect each other on many of these
faces, in an overall linear number of points. Equivalently, the
``danger'' is that the intersection circle $C_{ab}$ of a
corresponding pair of spheres $\bd{B_r(a)}, \bd{B_r(b)}$, for $a,b
\in P_R$, could intersect a linear number of faces of $\sigma_L$
(and each of these intersections is also an intersection point of
$\gamma_a$ and $\gamma_b$). See Figure~\ref{figure:new_case2_cubic}.

\begin{figure}[htbp]
\begin{center}

\input{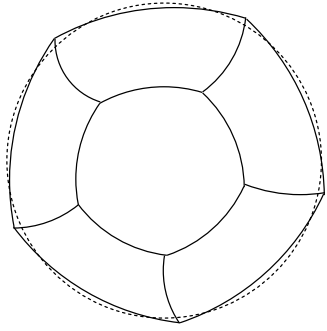}

\caption{\small \sf In a general setup (different than ours), an intersection circle of two balls (the dotted circle) may
intersect a linear number of faces of $\bd{K(P_L)}$.}
\label{figure:new_case2_cubic}
\end{center}
\end{figure}

\smallskip
\noindent{\bf Complexity of $M$.} Fortunately, in the assumed
configuration, this cubic behavior is impossible --- $C_{ab}$ can
meet only a constant number of faces of $\sigma_L$. Consequently,
the overall complexity of $M$ is only quadratic. This crucial claim
follows from the observation that, for $C_{ab}$ to intersect many
faces of $\sigma_L$, it must have many short arcs, each delimited by
two points on $\sigma_L$ and lying outside $K(P_L)$. The main
geometric insight, which rules out this possibility, and leads to
our improved algorithm, is given in the following lemma.

\begin{lemma}
\label{lemma:quadratic_K_P_L}
Let $\lambda$ be a $yz$-parallel plane, which separates $v_1$ from $c_1$.
Let $P_L \subseteq P$ be the subset of points of $P$ to the left of $\lambda$, and let $P_R = P \setminus P_L$.
Let $C_{ab}$ denote the intersection circle of $\bd{B_r(a)}, \bd{B_r(b)}$, for some pair of points $a, b \in P_R$, and let $q \in P_L$. If the arc $\omega = C_{ab} \setminus B_r(q)$ is smaller than a semicircle of $C_{ab}$, then at least one of its endpoints must lie to the right of $\lambda$.
\end{lemma}

\begin{proof}The situation and its analysis are depicted in Figure~\ref{figure:lemma_setup}. To slightly simplify the analysis, and without loss of generality, assume that $r = 1$. Let $h$ be the plane passing through $a$, $b$ and $q$. Let $c_{ab}$
denote the midpoint of $ab$, and let $w$ denote the center of the
circumscribing circle $Q$ of $\triangle qab$. Denote the distance $|ab|$
by $2x$, and the radius of $Q$ by $y$ (so $|wp_1|=|wp_2|=|wq|=y$).
Note that $c_{ab}$ and $w$ lie in $h$ and that $y \ge x$.
Observe that $c_{ab}$ is the center of the intersection circle $C_{ab}$ of
$\bd{B_r(a)}$ and $\bd{B_r(b)}$. See Figure~\ref{figure:lemma_setup}(a).

\begin{figure}[htbp]
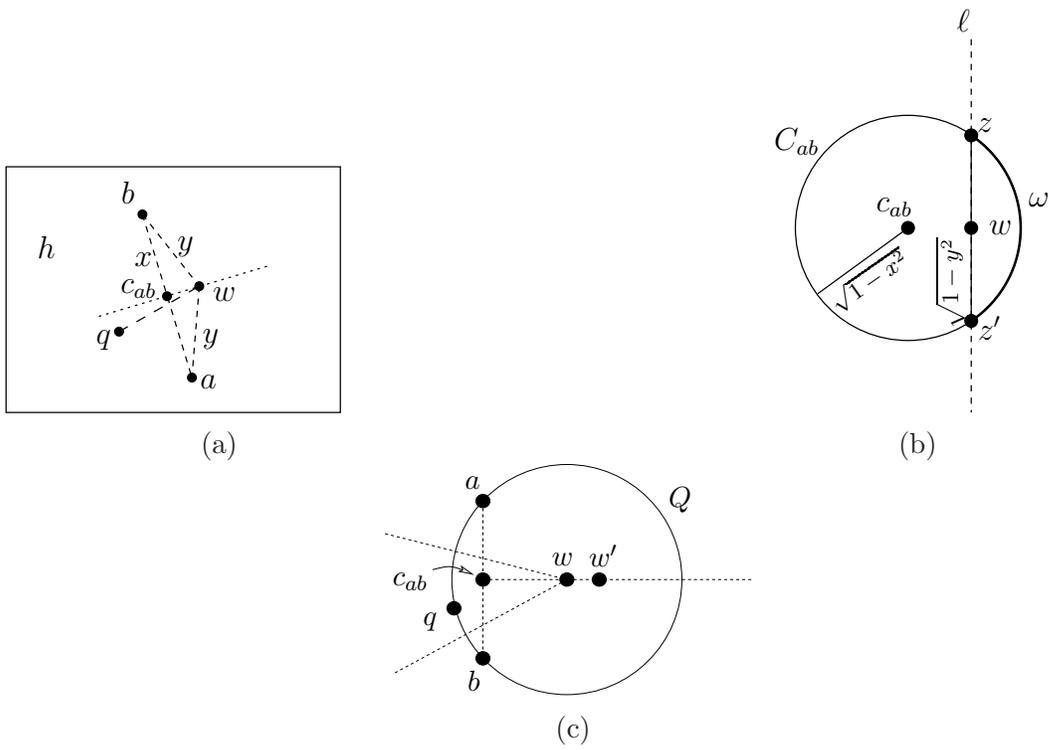

\begin{center}
\begin{tabular*}{0.8\textwidth}{c c c}
            \hspace{-35pt} {\input{h_setup_s2.pstex_t} } &  & {\input{C_ab_s2.pstex_t} } \\
    \small (a) & \hspace{15pt} & \small (b) \\
    \hspace{45pt} & {\input{abq_s2.pstex_t} } & \hspace{45pt}\\
     \hspace{-35pt} & \small (c) &
           \end{tabular*}

       \vspace{20pt}
\caption{\small \sf The setup in Lemma~\ref{lemma:quadratic_K_P_L}:
           (a) the setup within the plane $h$; (b) the setup within
    $C_{ab}$; (c) $ww'$ lies on the bisector of $ab$ in
           the direction that gets away from $q$.}

\label{figure:lemma_setup}
\end{center}
\end{figure}

The intersection points $z,z'$ of $C_{ab}$ and $\bd{B_r(q)}$ are the
intersection points of the three spheres
$\bd{B_r(a)}$, $\bd{B_r(b)}$, and $\bd{B_r(q)}$.
They lie on the line $\ell$ passing through $w$ and orthogonal
to $h$, at equal distances $\sqrt{1-y^2}$ from $w$. See Figure~\ref{figure:lemma_setup}(b).
(If $y>1$ then $z$ and $z'$ do not exist, in which case $C_{ab}$ does not intersect $\bd{B_r(q)}$; in what follows we assume
that $y\le 1$.)
Hence, within $C_{ab}$, $zz'$ is a chord of length $2\sqrt{1-y^2}$.
In the assumed setup, $z$ and $z'$
delimit a short arc $\omega$ of $C_{ab}$, which lies outside $B_r(q)$,
so points on the arc are (equally) closer to $a$ and $b$ than to
$q$.

Hence, the projection of the arc $\omega$ onto $h$ is a small
interval $ww'$, which lies on the bisector of $ab$ in the
direction that gets away from $q$; that is, it lies on the Voronoi
edge of $ab$ in the diagram $\Vor(\{a,b,q\})$ within $h$. See Figure~\ref{figure:lemma_setup}(c).
Moreover, $c_{ab}$ also lies on the bisector, but it has to lie on the
other side of $w$, or else the smaller arc $\omega$ would have to lie
inside $B_r(q)$. That is, $c_{ab}$ has to be closer to $q$ than to $a$
and $b$. Since $\lambda$ separates $a$ and $b$ from $q$, it also separates
 $c_{ab}$ from $q$. Moreover, the preceding arguments are easily seen to imply that $wq$ crosses $ab$ (as in Figure~\ref{figure:lemma_setup}(a)), which implies that $\lambda$ also separates $q$ and $w$, so
$w$ has to lie to the right of $\lambda$. Since $z$ and $z'$ lie on
two sides of $w$ on the line $\ell$, at least one of them has to lie on
the same side of $\lambda$ as $w$ (i.e., to the right of
$\lambda$). This completes the proof.
\end{proof}

Let $a,b \in P_R$ and consider those arcs of $C_{ab}$ which lie outside $K(P_L)$ but their endpoints lie on $\sigma_L$. Clearly, all these arcs are pairwise disjoint. At most one such arc can be larger than a semicircle. Let $\omega$ be an arc of this kind which is smaller than a semicircle, and let $q \in P_L$ be such that one endpoint of $\omega$ lies on $\bd B_r(q)$. Then $\omega' = C_{ab} \setminus B_r(q)$ is contained in $\omega$ and therefore is also smaller than a semicircle. By Lemma~\ref{lemma:quadratic_K_P_L}, exactly one endpoint of $\omega'$ lies to the right of $\lambda$ (the other endpoint lies on $\sigma_L$).
Note that $C_{ab}$ cannot have more than two such short arcs lying outside $K(P_L)$, since, due to the convexity of $C_{ab}$, only two arcs of $C_{ab}$ can have their two endpoints lying on opposite sides of $\lambda$. Hence the number of arcs of $C_{ab}$ under consideration is at most 3, implying that $\gamma_a$ and $\gamma_b$ intersect at most three times, and thus the complexity of $M$ is $O(n^2)$, as asserted.

\smallskip
\noindent{\bf Constructing and searching $M$.} The next step of the
algorithm is to compute $M$. We have already constructed
$\bd{K(P_L)}$, in $O(n \log n)$ time, and, in additional linear
time, we can compute its portion $\sigma_L$ to the left of $\lambda$
(we omit the straightforward details). We compute the intersection
curve $\gamma_p$ of $B_r(p)$ and $\sigma_L$, for each $p \in P_R$,
in $O(n \log n)$ time, by computing the intersection $K(P_L \cup
\{p\})$, and obtaining the curve which bounds the portion of the
unique face of $\bd{K(P_L \cup \{p\})}$ within $\sigma_L$. If
necessary, we also split $\gamma_p$ into portions, such that each
portion is contained in a different face of $\sigma_L$. The total
cost of computing all curves $\{\gamma_p \mid p \in P_R\}$, and
spreading them along the faces of $\sigma_L$, is $O(n^2 \log n)$.
Then, for each face $f$ of $\sigma_L$, we consider the portions of
all the arcs $\gamma_p$, for $p \in P_R$, within $f$, and compute
their arrangement (which is the portion of $M$ which lies in $f$).
To this end, we use standard line-sweeping~\cite{book}, to report
all the intersections of $n$ curves in the plane in $O((n+k) \log
n)$ time, where $k = k_f$ is the complexity of the resulting
arrangement on $f$. Hence, the total cost of computing the portion
of $M$ on all the faces of $\sigma_L$ is $\sum_{f \in \sigma_L}
O((n+k_f) \log n) = O(n^2 \log n) + O(\log n) \cdot \sum_{f \in
\sigma_L} k_f = O(n^2 \log n)$, since the complexity of $M$ is
$O(n^2)$.

We next perform a traversal of the cells of $M$ in a manner similar to the one used in Section~\ref{sec:cubic_alg}, via a tour, which proceeds from each visited cell to an adjacent one. For each cell $\tau$ that we visit, we place
the center $c_1$ of $B_1$ in $\tau$, and maintain dynamically the subset
$P_\tau^+$ of points of $P$ not covered by $B_1$. (Here, unlike the algorithm of Section~\ref{sec:cubic_alg}, the complementary set  $P_\tau^-$ is automatically covered by $B_1$ and there is no need to test it.) As before,
when we move
from one cell $\tau$ to an adjacent cell $\tau_1$, $P_{\tau_1}^+$ gains
one point or loses one point. This implies that this tour generates only $O(n^2)$ connected life-spans of the points of $P$, where a life-span of a point $p$ is a maximal connected interval of the tour, in which $p$ belongs to $P_\tau^+$. We can thus use a segment tree $T_M$ to store these life-spans, as before. Each
leaf $u$ of $T_M$ represents a cell $\tau$ of $M$, and the balls not containing $\tau$ are those with life-spans that are stored at the nodes on the path from the root to $u$. Since $M$ has a quadratic number of
cells, $T_M$ has a total
of $O(n^2)$ leaves. Arguing exactly as in Section~\ref{subsec:decision_procedure}, we can compute $T_M$ in overall $O(n^2 \log^2 n)$ time, and the total storage used by $T_M$ is $O(n^2 \log n)$.

As in Section~\ref{subsec:decision_procedure}, we next test, for
each leaf $u$ of $T_M$, whether the spherical polytopes along the
path from the root to $u$ have non-empty intersection. We do this
using the parametric search technique described in
Proposition~\ref{prop:intersection_test}, which takes $O(\log^5 n)$
time for each path, for a total of $O(n^2 \log^5 n)$. More
precisely, as above, we also need to distinguish between $r = r^*$
and $r > r^*$. We therefore stop only when both the intersection
along the path and the cell of $\sigma_L$ corresponding to $u$ are
non-degenerate, and then report that $r^* < r$. Otherwise, we
continue running the above procedure over all paths of $T_M$, and
repeat it for each of the $O(1/\beta^3)$ combinations of an
orientation $v$ and a separating plane $\lambda$. If we find at
least one (degenerate\footnote{Note that $\bigcap\{B_r(p) \mid p \in
P_\tau^-\}$ is non-degenerate if $\tau$ is a 2-face or an edge. If
$\tau$ is a vertex we test for degeneracy as in the procedure in
Section~\ref{subsec:decision_procedure}. Determining whether
$\bigcap\{B_r(p) \mid p \in P_\tau^+\}$ is degenerate is also
performed using that procedure.}) solution, we report that $r^* =
r$, and otherwise conclude that $r^*
> r$. Hence, the cost of handling Case~2, and thus also the overall
cost of the decision procedure, is $O((1/\beta^3) n^2 \log^5 n)$.

\subsection{Solving the optimization problem}
\label{sec:separated_centers_optimization}

We now combine the decision procedure $\Gamma$ described in Section~\ref{subsec:improved_decision_procedure} with the randomized optimization
technique of Chan~\cite{TCG} (as briefly described in Section~\ref{subsec:optimization_problem}),
to obtain a solution for the optimization problem.

The decision procedure $\Gamma$, on a specified radius $r$, relies on an apriori knowledge of a lower bound $\beta$ for the separation ratio
$|c_1c_2|/r$. To supply such a $\beta$, let $r_0$ denote the radius of the smallest
enclosing ball of $P$, and observe that if there exist two balls $B_1, B_2$ of radius $r$ covering $P$ then the smallest ball $B^*$ enclosing $B_1 \cup B_2$ must be at least as large as the smallest enclosing ball of $P$, so its radius must be at least $r_0$. Since this radius is
$(1+\beta/2)r$ (see Figure~\ref{figure:new_case2_SEB}), we have
$(1+ \beta/2)r \geq r_0$ or $\beta \geq 2(r_0/r -1)$.
It follows that the running time of the decision procedure $\Gamma$ is
$$O\left(\frac{1}{\beta^3} n^2 \log^5 n\right) =
O\left(\frac{1}{\left(1-r/r_0 \right)^3} n^2 \log^5 n\right).$$

\begin{figure}[htbp]
\begin{center}
\input{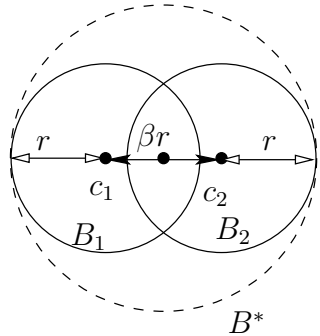}
\caption{\small \sf The smallest enclosing ball $B^*$ of $B_1 \cup
B_2$.} \label{figure:new_case2_SEB}
\end{center}
\end{figure}

Chan's technique starts with a very big $r$ (for all practical
purposes we can start with $r = r_0$) and shrinks it as it iterates
over the subproblems. Therefore, running Chan's technique in a
straightforward manner, starting with $r = r_0$, will make it
potentially very inefficient, because the initial executions of
$\Gamma$, when $r$ is still close to $r_0$, may be too expensive due
to the large constant of proportionality (not to mention the run at
$r_0$ itself, which the algorithm cannot handle at all). We need to
fine-tune Chan's technique, to ensure that we do not consider values
of $r$ which are too close to $r_0$. To do so, we consider the
interval $(0, r_0)$ which contains $r^*$, and run an ``exponential
search'' through it, calling $\Gamma$ with the values $r_i = r_0
\left(1-1/2^i\right)$, for $i = 1,2, \ldots$, in order, until the
first time we reach a value $r' = r_i \geq r^*$. Note that $1-
r'/r_0 = 1/2^i$ and $1/2^i < 1- r^*/r_0 < 1/2^{i-1}$, so our lower
bound estimates for the separation ratio $\beta$ at $r'$ and at
$r^*$ differ by at most a factor of $2$, so the cost of running
$\Gamma$ at $r'$ is asymptotically the same as at $r^*$. Moreover,
since the (constants of proportionality in the) running time bounds
on the executions of $\Gamma$ at $r_1,\ldots,r_i$ form a geometric
sequence, the overall cost of the exponential search is also
asymptotically the same as the cost of running $\Gamma$ at $r^*$. We
then run Chan's technique, with $r'$ as the initial minimum radius
obtained so far. Hence, from now on, each call to $\Gamma$ made by
Chan's technique will cost asymptotically no more than the cost of
calling $\Gamma$ with $r'$ (which is asymptotically the same as
calling $\Gamma$ with $r^*$).

\paragraph{Combining Chan's technique with the decision procedure $\Gamma$.}
To apply Chan's technique with our decision procedure, we use the
same cutting-based decomposition as in
Section~\ref{subsec:optimization_problem}. That is, we replace each
point $p \in P$ by its dual plane $p^* \in P^*$, and construct a
$(1/\varrho)$-cutting of $\A(P^*)$, for some sufficiently large
constant parameter $\varrho>0$. We then apply Chan's technique to
the resulting subproblems (where each subproblem corresponds to a
simplex $\Delta_i$ of the cutting), using the improved decision
procedure $\Gamma$ on each of them, and recursing into some of them,
as required by the technique. As in Section~\ref{sec:cubic_alg}, the
recursion and the application of the decision procedure are not
``pure'', because they need to consider also those planes that miss
the current simplex. (Note that in the problem decomposition we use,
for simplicity, the full 3-dimensional arrangement $\A(P^*)$, of
cubic size. This, however, does not affect the asymptotic running
time, because we have only a constant number of subproblems, and
Chan's technique recurses into only an expected logarithmic number
of them.) Given a radius $r$, we compute the lower bound $\beta =
2\left(\frac{r_0}{r} -1\right)$ for the separation ratio
$\frac{|c_1c_2|}{r}$, where $c_1, c_2$ are the centers of the two
covering balls, as above. Consider the application of $\Gamma$ to a
subproblem represented by a simplex $\Delta_i$ of the cutting. The
presence of ``global'' points (those dual to planes passing above or
below $\Delta_i$) forces us, as in
Section~\ref{subsec:optimization_problem}, to modify the ``pure''
version of $\Gamma$ described above. We use the same notations as in
Section~\ref{sec:cubic_alg}.

\begin{figure}[htbp]
\begin{center}
\input{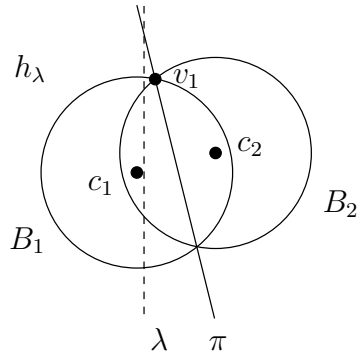}

\caption{\small \sf $h_\lambda$ does not contain any point of
$P_{\Delta_i}^+$.} \label{figure:h_lambda}
\end{center}
\end{figure}

We again rotate the coordinate axes, in $O\left(1/\beta^2\right)$
ways (in the same manner as in the ``pure'' decision procedure), and
draw $O(1/\beta)$ $yz$-parallel planes, such that, at the correct
orientation, one of these planes, $\lambda$, separates $c_1$ from
$v_1$ (if there is a solution for $r$). As in the pure case, we may
assume that the $x$-span of $P$ is at most $5r$; a larger span is
handled earlier. We assume, without loss of generality, that
$P_{\Delta_i}^- \subseteq B_1$, and that $P_{\Delta_i}^+ \subseteq
B_2$. Recall also that the points in the left halfspace $h_\lambda$
bounded by $\lambda$ are all contained in $B_1$. Moreover, the plane
$\pi$ containing the intersection circle $C_{12}$ is dual to a point
$\pi^*$, which has to separate $(P^*)_{\Delta_i}^+$ from
$(P^*)_{\Delta_i}^-$. Hence, all the points of $P_{\Delta_i}^+$ have
to lie on the other side of $\pi$, and in $B_2$, which is easily
seen to imply that none of them can lie in $h_\lambda$. See
Figure~\ref{figure:h_lambda}. We thus verify that $P_{\Delta_i}^+
\cap h_\lambda = \emptyset$, aborting otherwise the guess of
$\lambda$. (Note that, in contrast, points of $P_{\Delta_i}^-$ can
also lie to the right of $\lambda$.)

We now have a subset $P_L \subseteq P_{\Delta_i}^0$ of $O(m)$ points
to the left of $\lambda$, which are assumed, together with the
points of $P_{\Delta_i}^-$, to be contained in $B_1$. Note however
that, for Lemma~\ref{lemma:quadratic_K_P_L} to hold, we have to
define $\sigma_L$ only in terms of the points to the left of
$\lambda$. Therefore, we compute the surface $\sigma_L' = \bd{K(P_L
\cup (P_{\Delta_i}^- \cap h_\lambda)) \cap h_\lambda}$ and search on
it for a placement of the center $c_1$ of $B_1$. However, since the
remaining points of $P_{\Delta_i}^-$ are also assumed to belong to
$B_1$, we need to consider only the portion of $\sigma_L'$ inside
$\bigcap\{B_r(p) \mid p \in P_{\Delta_i}^- \setminus h_\lambda\}$.
Let $\sigma_L''$ denote this portion. It is easy to compute
$\sigma_L''$ in $O(n \log n)$ time. It is easily checked that $c_1$
must lie on $\sigma_L''$ (if there is a solution for the current
situation). So far, the cost of the decision procedure also depends
(cheaply --- see below) on the initial input size $n$, but the
saving in this setup comes from the fact that it suffices to
intersect the $O(m)$ spheres $\bd{B_r(p)}$, for $p \in
P_{\Delta_i}^0 \setminus h_\lambda$, with $\sigma_L''$ to obtain the
map $M$, since only the points of $P_{\Delta_i}^0$ are
``undecided''. (The points of $P_{\Delta_i}^+$ are always placed in
$B_2$ as already discussed.)

Note that $\sigma_L''$ need not to be connected, so it may seem
impossible to visit all the cells of $M$ in a single connected tour.
Nevertheless, we will be able to do it, in a manner detailed below.
We thus build a segment tree $T_M$ to maintain the subset $P'(c_1)$
of points of $P$ not covered by $B_1$. We build and query $T_M$ as
is done in Section~\ref{subsec:decision_procedure}, except for the
following modifications. First, note that the points of
$P_{\Delta_i}^+$ are assumed to be contained in $B_2$. Thus, the
points of $P_{\Delta_i}^+$, that in the decision procedure were
considered in building $M$, do not need to be considered as part of
$M$ now, rather it is enough to build the spherical polytope
$\bigcap \{B_r(p) \mid p \in P_{\Delta_i}^+ \}$ and place it at the
root of $T_M$. Second, we claim that $M$ is of complexity $O(m n)$.
To see this, let $\C^0$ denote the set of curves $\{\bd{B_r(p)} \cap
\sigma_L'' \mid p \in P_{\Delta_i}^0\}$. Each pair of curves of
$\C^0$ can intersect each other in only a constant number of points,
as proved in Section~\ref{subsec:improved_decision_procedure}.
Hence, the complexity of the arrangement of the $O(m)$ curves in
$\C^0$, formed on $\sigma_L''$, is $O(m^2)$. However, $\sigma_L''$
itself is of complexity $O(n)$, and each edge of $\sigma_L''$ may
intersect the curves of $\C^0$ at $O(m)$ points. Hence, the
complexity of the map $M$ is $O(m n)$, but the number of its
vertices that lie in the interior of the faces of $M$ is only
$O(m^2)$.

To overcome the possible disconnectedness of $\sigma_L''$, we
proceed as follows. We consider the (connected) network of the
$O(n)$ edges of $\sigma_L'$, and intersect each of these edges with
the $m$ balls $B_r(p)$, for $p \in P_{\Delta_i}^0$. We construct a
tour of this network, which visits $O(mn)$ arcs along the edges of
$\sigma_L'$, and append to this ``master tour'' separate tours of
each face of $\sigma_L''$. We get in this way a single grand tour of
the cells of $M$ (which also traverses some superfluous arcs of
$\sigma_L' \setminus \sigma_L''$), of length $O(mn)$, which has the
incremental property that we need: Moving from any cell or arc of
the tour to a neighbor cell or arc incurs an insertion or a deletion
of a single point into/from $P'(c_1)$.

\paragraph{Running time.}
For each cell of $M$ we run the procedure described in
Proposition~\ref{prop:intersection_test} for determining whether the
intersection of the corresponding spherical polytopes is nonempty
(and whether it has nonempty interior). Therefore, solving each
subproblem requires $O(m n \log^5 n)$ time. The $O(m n \log n)$ time
required to build $M$, and the $O(n \log n)$ time required to
construct the intersection of the balls in $\{B_r(p) \mid p \in
P_{\Delta_i}^+\}$, are all subsumed in that cost. Repeating this for
each of the $O(1/\beta^3)$ guesses of an orientation and a
separating plane, results in $O\left((1/\beta^3) mn \log^5 n\right)$
rnning time. When the recursion bottoms out, we handle it the same
way as in Section~\ref{subsec:optimization_problem}.

Arguing similarly to the less efficient solution, we obtain the
following recurrence for the maximum expected cost $T(m, n)$ of
solving a recursive subproblem involving $m$ ``local'' points, where
$n$ is the number of initial input points in $P$.

\begin{equation}
\label{eqn:T(m,n)_2} T(m,n) \leq \left\{ \begin{array}{ll}
\ln(c\varrho^3)T(m/\varrho,n) + O\left((1/\beta^3) m n \log^5 n\right),   & \mbox{for $m \geq \varrho$,}\\
O(n),                                        & \mbox{for $m <
\varrho$,}
\end{array}\right.
\end{equation}
where $c$ is an appropriate absolute constant (as in
Section~\ref{subsec:optimization_problem}), $\varrho$ is the
parameter of the cutting, chosen to be a sufficiently large constant
(to satisfy (\ref{equation:lnr}), as above, with $\gamma = 2$), and
$\beta = 2\left(r_0/r'-1\right)$, where $r'$ is the value of $r$ at
which the initial exponential search is terminated.

It can be shown rather easily (and we omit the details, as we did in
the preceding section), that the recurrence~(\ref{eqn:T(m,n)_2})
yields the overall bound $O\left((1/\beta^3) n^2 \log^5 n\right)$ on
the expected cost of the initial problem; i.e.,
$$T(n,n) = O\left((1/\beta^3) n^2 \log^5 n\right).$$ We thus finally obtain
our main result:
\begin{theorem}
Let $P$ be a set of $n$ points in $\mathbb{R}^3$. A 2-center for $P$
can be computed in $O((n^2 \log^5 n)/(1-r^*/r_0)^3 )$ randomized
expected time, where $r^*$ is the radius of the balls of the
2-center for $P$ and $r_0$ is the radius of the smallest enclosing
ball of $P$.
\end{theorem}

\section{Efficient Emptiness Detection of Intersection of Spherical
Polytopes} \label{sec:spherical_polytopes} In this section we
describe an efficient procedure for testing emptiness (and
non-degeneracy) of the intersection of spherical polytopes, as
prescribed in Proposition~\ref{prop:intersection_test}. Let $\S$ be
a collection of spherical polytopes, each defined as the
intersection of at most $n$ balls of a fixed radius $r$. Fix a
spherical polytope $S \in \S$. To simplify the forthcoming analysis,
we assume that the centers of the balls involved in the polytopes of
$\S$ are in general position, meaning that no five of them are
co-spherical, and that there exists at most one quadruple of centers
lying on a common sphere of radius $r$. As is well known, each ball
$b$ participating in the intersection $S$ contributes at most one
(connected) face to $\bd{S}$ (see~\cite{ER}). The vertices and edges
of $S$ are the intersections of two or three bounding spheres,
respectively (at most one vertex might be incident to four spheres).
Hence $\bd{S}$ is a planar (or, rather, spherical) map with at most
$|S|$ faces, which implies that the complexity of $\bd{S}$ is
$O(|S|)$.

We preprocess $S$ into a point-location structure. We first
partition $\bd{S}$ into its upper portion $\bd{S}^+$ and lower
portion $\bd{S}^-$. We project vertically each of $\bd{S}^+$ and
$\bd{S}^-$ onto the $xy$-plane and obtain two respective planar maps
$M^+$ and $M^-$ (see Figure~\ref{figure:projection}). For each face
$\zeta$ of each map we store the ball $b$ that created it; that is,
$\zeta$ is the projection of the (unique) face of $\bd{S}$ that lies
on $\bd{b}$. The $xy$-projection $S^*$ of $S$ is equal to both
projections of $\bd{S^+}$, $\bd{S^-}$, and is bounded by a convex
curve $E^*$ that is the concatenation of the $xy$-projections of
certain edges of $S$ and of portions of horizontal equators of some
of its balls.

\begin{figure}[htbp]
\begin{center}
\input{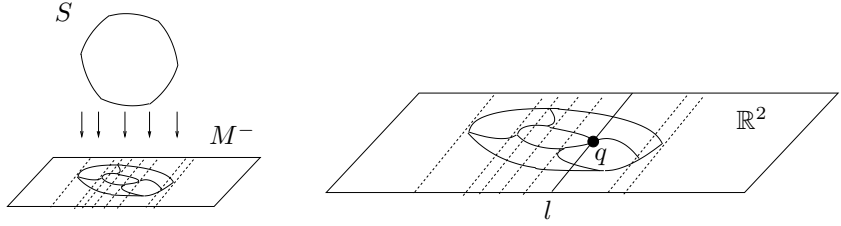}
\caption{\small \sf Projecting $\bd{S_i}^-$ vertically onto the
$xy$-plane (left), and the point location structure for the
resulting map $M_i^-$ (right).} \label{figure:projection}
\end{center}
\end{figure}

We apply the standard point-location algorithm of Sarnak and
Tarjan~\cite{ST} to each of the maps $M^+, M^-$. That is, we divide
each planar map into slabs by parallel lines (to the $y$-axis)
through each of the endpoints (and locally $x$-extremal points) of
the arcs obtained by projecting the edges of $\bd{S}$, including the
new equatorial arcs. Using the persistent search structure
of~\cite{ST}, the total storage is linear in $|S|$ and the
preprocessing cost is $O(|S| \log |S|)$, where $|S|$ is the number
of balls forming $S$. To locate a point $q_0$ in $M^+$ (or in
$M^-$), we first find the slab in the $x$-structure that contains
$q_0$, and then find the two curves between which $q_0$ lies in the
$y$-structure.\footnote{\small All these standard details are
presented to make more precise the infrastructure used by the
higher-dimensional routines $\Pi_1$ and $\Pi_2$.}

To determine whether $q \in S^*$, we locate the face $\zeta^+$
(resp., $\zeta^-$) of the map $M^+$ (resp., $M^-$) that contains
$q$, as just described. Each of these faces can be a 2-face, an edge
or a vertex. We therefore retrieve a set $\B^+$ (resp., $\B^-$) of
the one, two, or three or four balls associated (respectively) with
the 2-face, edge or vertex containing $q$. (We omit here the easy
construction of witness balls when the faces $\zeta^+$ and $\zeta^-$
are not associated with any ball, that is, $q \notin S^*$.)

Let $\B$ denote the set $\B^+ \cup \B^-$. We observe that $q \in
S^*$ if and only if the $z$-vertical line $\lambda_q$ through $q$
intersects $S$. Moreover, we have, by construction, $\lambda_q \cap
S = \lambda_q \cap (\bigcap \B)$. Hence $q \in S^*$ if and only if
$s \coloneqq \lambda_q \cap (\bigcap \B) \neq \emptyset$. Clearly,
if we put $N = \sum_{S \in \S} |S|$, then the preprocessing stage
takes a total of $O(N \log n)$ time and requires $O(N)$ storage.

Next, let $S_1,\ldots,S_t$ be $t \leq \log n$ spherical polytopes of
$\S$, for which we want to determine whether $K = \bigcap_{i=1}^t
S_i$ is nonempty (and, if so, whether it has nonempty interior). We
solve this problem by employing a technique similar to the
multi-dimensional parametric searching technique of
Matou\v{s}ek~\cite{JM} (see also~\cite{AES, TCA, CMS, NPT}). We
solve in succession the following three subproblems, $\Pi_0(q)$,
where $q$ is a point in the $xy$-plane, $\Pi_1(l)$, where $l$ is a
$y$-parallel line in the $xy$-plane, and $\Pi_2$, over the entire
$xy$-plane. In the latter problem we wish to to determine whether
the $xy$-projection $K^*$ of $K$ is nonempty. During the execution
of the algorithm for solving $\Pi_2$, we call recursively the
algorithm for solving $\Pi_1(l)$, for certain $y$-parallel lines $l
\subset \mathbb{R}^2$, and we wish to determine whether $K^*$ meets
$l$. If so, then $\Pi_2$ is solved directly (with a positive
answer). Otherwise, we wish to determine which side of $l$, within
$\mathbb{R}^2$, can meet $K^*$ (since $K^*$ is convex, there can
exist at most one such side). The recursion bottoms out at certain
points $q \in l$, on which we run $\Pi_0(q)$ to determine whether
$K^*$ contains $q$. If so, then $\Pi_1(l)$ is solved directly (with
a positive answer). Otherwise, we determine which side of $q$,
within $l$, can meet $K^*$, and continue the search accordingly.

Our solutions to the subproblems $\Pi_k$, $0\leq k\leq 2$, are based
on generic simulations of the standard point-location machinery of
Sarnak and Tarjan~\cite{ST} mentioned above. In each of the
subproblems, if we find a point in $f \cap K^*$, for the respective
point, line, or the entire $xy$-plane $f$, we know that $K \neq
\emptyset$ and stop right away. If $f \cap K^* = \emptyset$, we want
to ``prove'' it, by returning a small set of \emph{witness balls}
$b_1,\ldots,b_y$, where, for each $j$, $b_j$ is one of the balls
that participates in some spherical polytope $S_i$ (so $b_j
\supseteq S_i$), so that their intersection $K_0 = \bigcap_{j=1}^y
b_j$ satisfies $f \cap K_0^* = \emptyset$ (where, as above, $K_0^*$
is the $xy$-projection of $K_0$). If $K_0 = \emptyset$ then $K =
\emptyset$ too and we stop. Otherwise (when $f$ is a line or a
point), $K_0$ determines the side of $f$ (within $\mathbb{R}^2$ if
$f$ is a line, or within the containing line $l$ if $f$ is a point)
that might meet $K^*$; the opposite side is asserted at this point
to be disjoint from $K^*$. We use this information to perform binary
search (or, more precisely, parametric search) to locate $K^*$
within the flat, from which we have recursed into $f$.  The
execution of the algorithm for solving $\Pi_2$ will therefore either
find a point in $K$ or determine that $K = \emptyset$, because it
has collected a small (as we will show, polylogarithmic) number of
witness balls, whose intersection, which has to contain $K$, is
found to be empty.

\paragraph{Solving $\Pi _0(q)$ for a point $q$.}
\label{subsec:Pi_0}
Here we have a point $q \in \mathbb{R}^2$ and we wish to determine
whether $q \in K^*$. To do so, we locate $q$ in each of the maps
$M_i^+$ (the $xy$-projection of $\bd S_i^+$) and $M_i^-$ (the
$xy$-projection of $\bd S_i^-$), for each $i=1,\ldots,t$. If $q$
lies outside the projection of at least one polytope $S_i$ then $q
\notin K^*$, and we return the witness balls that prove that $q
\notin S_i^*$. Otherwise, as explained above, each point location
returns a set $\B_i$ of $O(1)$ witness balls for $S_i$. We compute
the $t$ line segments $s_i = \lambda_q \cap (\bigcap\B_i)$, for each
$i = 1,\ldots,t$, where $\lambda_q$ is, as above, the $z$-vertical
line through $q$. We then have $K_0 \coloneqq \lambda_q \cap K =
\bigcap_{i=1}^t s_i$, so it suffices to compute this intersection
(in $O(t)$ time) and test whether it is nonempty. If $K_0$ is
nonempty, then we have found a point $q'$ in $K$. Otherwise, we
return the set $\B_0 = \bigcup \{\B_i \mid 1 \leq i \leq t\}$ of up
to $5 \log n$ balls as witness balls for the higher-dimensional step
(involving the $y$-parallel line containing $q$).

The time complexity for solving $\Pi_0(q)$ is $O(\log^2 n)$, since
it takes $O(\log n)$ time to compute, for each of the $O(\log n)$
spherical polytopes $S_i$, the intersection $\lambda_q \cap S_i$.

\paragraph{\bf Solving $\Pi _1(l)$ for a line $l$.}
\label{subsec:Pi_1}
Here we have a $y$-parallel line $l \subset \mathbb{R}^2$ and we
wish to determine whether $K^*$ meets $l$. We first locate $l$ in
each of the planar maps $M_i^+$ and $M_i^-$ of each $S_i$, and find
the slabs $\psi_i^+$ and $\psi_i^-$, which contain $l$ (in some
cases $l$ is the common bounding line of two adjacent slabs
$\psi_i'$ and $\psi_i''$ of $M_i^+$ or of $M_i^-$, so we retrieve
both slabs). We then run a binary search through the $y$-structure
of each of the obtained slabs to find a point in $K^* \cap l$, if
one exist. In each step of the search, within some fixed slab
$\psi_0$, we consider an arc $\gamma$ of the $y$-structure, and
determine whether $K^*$ meets $l$ above or below $\gamma$ (within
$\mathbb{R}^2$), assuming $K^* \cap l \neq \emptyset$. To this end,
we find the intersection point $q_0 = l \cap \gamma$, and run the
algorithm for solving $\Pi_0(q_0)$ (see Figure~\ref{figure:pi_1}).
If $q_0 \in K^*$, then we have found a point $q'$ in $K$, and we
immediately stop. Otherwise, we have a set $\B_0$ of up to $5 \log
n$ balls returned by the algorithm for solving $\Pi_0(q_0)$. We test
whether the $xy$-projection $K_0^*$ of $\bigcap \B_0$ intersects
$l$. If $K_0^* \cap l = \emptyset$, then (due to the convexity of
$K$) we know which side of $l$ (within $\mathbb{R}^2$) meets $K^*$,
and we return $\B_0$ as a set of witness balls for the
higher-dimensional (planar) step. Otherwise (again due to the
convexity of $K$), we know which side of $\gamma$, within $l$, meets
$K^*$, and we continue the search through the $y$-structure of
$\psi_0$ on this side. We continue the search in this manner, until,
for each $S_i$, we obtain an interval $\xi_i$ of $l$ between two
consecutive arcs of the $y$-structure of $\psi_0$, which meets $K^*$
(assuming $K^* \cap l \neq \emptyset$). Let $\Xi$ denote the
collection of all these intervals. Clearly, $K^* \cap l \subseteq
\bigcap \Xi$. We find the lowest endpoint $E^-$ among the top
endpoints of the intervals in $\Xi$ and the highest endpoint $E^+$
among the bottom endpoints of the intervals in $\Xi$, and test
whether $E^-$ is above $E^+$. If so, we consider the set $\B_1$ of
up to $10 \log n$ witness balls returned by the algorithms for
solving $\Pi_0(E^-)$ and $\Pi_0(E^+)$. If the $xy$-projection
$K_1^*$ of $\bigcap \B_1$ intersects $l$, then $K^*$ meets $l$ and
we stop immediately, for we have found that $K$ is nonempty.
Otherwise, we know which side of $l$ (within $\mathbb{R}^2$) can
meet $K^*$, and we return $\B_1$ as a set of witness balls for the
higher (planar) recursive level. If $E^-$ is not above $E^+$, then
$K^* \cap l = \emptyset$ and we return $\B_1$ as a set of witness
balls for the higher (planar) recursive level as well.\footnote{With
some care, the number of witness balls can be significantly reduced.
We do not go into this improvement, because handling the witness
balls is an inexpensive step, whose cost is subsumed by the cost of
the other steps of the algorithm.}

\begin{figure}[htbp]
\begin{center}
\input{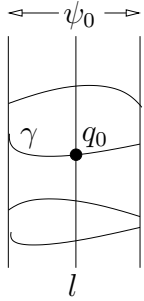}
\caption{\small \sf The line $l$ on which we run $\Pi_1(l)$. The
point $q_0$ on which we run $\Pi_0(q_0)$ is the intersection point
of $l$ with some arc $\gamma$.} \label{figure:pi_1}
\end{center}
\end{figure}

A naive implementation of the above procedure takes $O(\log^4 n)$
time, since for each of the $O(\log n)$ spherical polytopes $S_i$ we
run a binary search through the $y$-structure of at most two slabs
of each of the maps $M_i^+$ and $M_i^-$, and in each of the binary
search steps, we run the algorithm for solving $\Pi_0(q_0)$ for some
point $q_0$. The other substeps take less time. However, we can
improve the running time by implementing it in a parallel manner and
simulating the parallel version sequentially with a smaller number
of calls to $\Pi_0$.

We only parallelize the binary searches through the $y$-structure of
each $M_i^+$ and $M_i^-$, since the other substeps take less time.
To this end, we use $O(\log n)$ processors, one for each of the
planar maps $M_i^+$ and $M_i^-$, and we run in parallel the binary
search through the $y$-structure of each planar map using $O(\log
n)$ parallel steps. In each parallel step we need to ``compare''
$O(\log n)$ arcs with $K^*$ (one arc for each of the planar maps
$M_i^+$, $M_i^-$). We therefore intersect each such arc with $l$ and
obtain a set $Q$ of $O(\log n)$ intersection points. We then run a
binary search through the points of $Q$ (to locate $K^*$) using
$\Pi_0$. This determines the outcome of the comparisons of each of
the arcs with $K^*$, and the parallel execution can proceed to the
next step. Applying this approach to each of the $O(\log n)$
parallel steps results in an $O(\log^3 n \log \log n)$-time
algorithm for solving $\Pi_1(l)$. However, we can slightly improve
this bound further using a simple variant of Cole's
technique~\cite{RC}. More precisely, in each parallel step we have a
collection $Q$ of $O(\log n)$ weighted points, one for each map,
which we need to compare with $K^*$. We select the (weighted) median
point $q_0$ of $Q$ and run $\Pi_0(q_0)$. This determines the
outcomes of the comparisons between $K^*$ and each of the points in
$Q$ which lie to the opposite side of $q_0$ to the side containing
$K^*$. Points in $Q$ which lie in the same side of $q_0$ as $K^*$,
in level $j$ of the parallel implementation, are given weight
$1/4^{j-1}$ and we try to resolve their comparison to $K^*$ in the
next step. An easy calculation (simpler than the one used by Cole)
shows that this method adds only $O(\log n)$ steps to the $O(\log
n)$ parallel steps of the searches, and now in each parallel step we
perform only one call to $\Pi_0$ (see~\cite{RC} for more details).
Therefore, the total running time of $\Pi_1(l)$ is $O(\log^3 n)$.

\paragraph{\bf Solving $\Pi _2$.}
\label{subsec:Pi_2}
We next consider the main problem $\Pi _2$, where we want to
determine whether $K^* \neq \emptyset$ (i.e., whether $K \neq
\emptyset$). We use parametric searching, in which we run the point
location algorithm that we used for solving $\Pi_0$, in the
following generic manner.

In the first stage of the generic point location, we run a binary
search through the slabs of each of the planar maps $M_i^+$ and
$M_i^-$, for $i=1,\ldots,t$. In each step of the search through any
of the maps, we take a line $l_0$ delimiting two consecutive slabs
of the map, and run the algorithm for solving $\Pi _1(l_0)$, thereby
deciding on which side of $l_0$ to continue the search. At the end
of this stage, unless we have already found a point in $K$ or
determined that $K$ is empty, we obtain a single slab in each map
that contains $K^*$.  Let $\psi$ denote the intersection of these
slabs, which must therefore contain $K^*$ (unless $K$ is empty). The
cost of this part of the procedure is $O(\log^5 n)$.

In the next stage of the generic point location, we consider each
map $M_i^+$ or $M_i^-$ (for simplicity we refer to it just as $M_i$)
separately, and run a binary search through the $y$-structure of its
slab $\psi_i$ that contains $\psi$. In each step of the search we
consider an arc $\gamma$ of the $y$-structure, and determine which
side of $\gamma$ (within the slab $\psi$), can meet $K^*$, assuming
that $\psi \cap K^* \neq \emptyset$; if $\gamma \cap K^* \neq
\emptyset$ we will detect it and stop right away. Before describing
in detail how to resolve each comparison with an arc $\gamma$, we
note that this results in $O(\log n)$ comparisons of arcs $\gamma$
to $K^*$ for each of the $O(\log n)$ planar maps $M_i^+$ and
$M_i^-$. However, we can reduce the number of comparisons to $O(\log
n)$ in total, by simulating (sequentially) a parallel implementation
of this step, as follows. There are $O(\log n)$ parallel steps, and
in each step we execute a single step of the binary search in each
of the maps $M_i^+, M_i^-$. In each parallel step we need to compare
$K^*$ to a set $G$ of $O(\log n)$ arcs, one of each of the planar
maps $M_i^+, M_i^-$. Consider the portion $\A'(G)$ of the
arrangement $\A(G)$ of the arcs in $G$ which lies in $\psi$. Let
$L(G)$ denote the set of $O(\log^2 n)$ $y$-parallel lines which pass
through the vertices of $\A'(G)$. We run a binary search through the
lines of $L(G)$, using calls to the algorithm for $\Pi_1$ to guide
the search, to locate $K^*$ amid these lines, in a total of
$O(\log^3 n \log \log n)$ running time. This step (if it did not
find a line crossing $K^*$) may trim $\psi$ to a narrower slab
$\psi'$ in which $K^*$ must lie if $K^* \neq \emptyset$. Put $G' =
\{\gamma \cap \psi' \mid \gamma \in G\}$, and observe that the arcs
of $G'$ are pairwise disjoint and form a sorted sequence in the
$y$-direction. We then perform a binary search through the arcs in
$G'$, using $O(\log \log n)$ comparisons to $K^*$. Each comparison
is carried out in $O(\log^4 n)$ time, in a manner detailed below.
Once the binary search is terminated, we can determine the outcomes
of the comparisons between $K^*$ and each of the arcs in $G'$ and
proceed to the next parallel step. Applying this approach to each of
the $O(\log n)$ parallel steps results in an $O(\log^5 n \log \log
n)$-algorithm for solving $\Pi_2$. We again use an appropriate
variant of Cole's technique to improve the running time by a $\log
\log n$ factor, in a manner similar to the one described in the
solution of $\Pi_1$.

To carry out a comparison between an arc $\gamma \in G'$ and $K^*$,
we act under the assumption that $\gamma \cap K^* \neq \emptyset$,
and try to locate a point of $\gamma \cap K^*$ in each of the other
maps. Suppose, to simplify the description, that we managed to
locate the entire $\gamma$ in a single face of each of the other
maps $M_j^+$, $M_j^-$. This yields a set $\B$ of $O(t)$ balls, so
that a point $v \in \gamma$ lies in $K^*$ if and only if it lies in
the $xy$-projection $K_0^*$ of $\bigcap \B$. We then test whether
$\gamma$ intersects $K_0^*$. If so, we have found a point in $K$ and
stop right away. Suppose then that $K_0^* \cap \gamma = \emptyset$.
If $K_0^* \cap \psi' = \emptyset$ then $K$ must be empty, because we
already know that $K^* \subset \psi'$. If $K_0^* \cap \psi' \neq
\emptyset$, then we know on which side of $\gamma$ to continue the
binary search in (the portion within $\psi'$ of) $\psi_i$.

\begin{figure}[htbp]
\begin{center}
\input{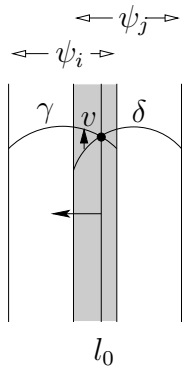}
\caption{\small \sf Comparing $\gamma \cap K^*$ with $\delta$. The
outcome of $\Pi_1(l_0)$ determines (a) the side of $\delta$ in which
the search in $\psi_j$ should continue, and (b) the portion of
$\gamma$ which can still meet $K^*$. The subslab $\psi'$ is drawn
shaded.} \label{figure:pi_2}
\end{center}
\end{figure}

In general, though, $\gamma$ might split between several cells of a
map $M_j$, where $M_j$ denotes, as above, one of the maps $M_j^+$ or
$M_j^-$. This forces us to narrow the search to a subarc of
$\gamma$, in the following manner. We run a binary search through
the $y$-structure of the corresponding slab $\psi_j$ of $M_j$, which
contains $\psi'$, and repeat it for each of the maps $M_j$. In each
step of the search, we need to compare $\gamma$ (or, more precisely,
some point in $\gamma \cap K^*$) with some arc $\delta$ of $\psi_j$,
which we do as follows. If $\gamma$ lies, within $\psi'$, completely
on one side of $\delta$, we continue the binary search in $\psi_j$
on that side of $\delta$. If $\gamma$ intersects $\delta$, we pick
an intersection point $v$ of $\gamma$ and $\delta$, pass a
$y$-parallel line $l_0 \subset \mathbb{R}^2$ through $v$, and run
the non-generic version of the algorithm to solve $\Pi_1(l_0)$. (See
Figure~\ref{figure:pi_2}.) As before, if $l_0 \cap K^* \neq
\emptyset$ we detect this and stop. Otherwise, we know which of the
two portions of $\gamma$, delimited by $v$, can intersect $K^*$. We
repeat this step for each of the at most four intersection points of
$\gamma$ and $\delta$ (observing that these are elliptic arcs), and
obtain a connected portion $\gamma'$ of $\gamma$, delimited by two
consecutive intersection points, whose relative interior lies
completely above or below $\delta$, so that $\gamma \cap K^*$, if
nonempty, lies in $\gamma'$. This allows us to resolve the generic
comparison with $\delta$, and continue the binary search through
$\psi_j$. (On the fly, each comparison with a line $l_0$ narrows
$\psi'$ still further.)

To make this procedure more efficient, we perform the binary
searches through the slabs $\psi_j$ in parallel, as follows. As
before, we run in parallel the binary searches through each of the
slabs $\psi_j$ using $O(\log n)$ parallel steps. In each parallel
step we need to compare a set $D$ of $O(\log n)$ arcs to $\gamma$,
one arc $\delta$ from each planar map $M_j$. We intersect each of
the arcs in $D$ with $\gamma$ and obtain a set $Z$ of $O(\log n)$
intersection points. Let $L_Z$ denote the set of the $O(\log n)$
$y$-parallel lines which pass through the points of $Z$. We run a
binary search through the lines of $L_Z$, using calls to the
algorithm for $\Pi_1$ to guide the search, in a total of $O(\log^3 n
\log \log n)$ running time. We obtain a connected portion $\gamma'$
of $\gamma$, delimited by two consecutive intersection points of
$Z$, whose relative interior lies completely above or below each
$\delta \in D$, so that $\gamma \cap K^*$, if nonempty, lies in
$\gamma'$. This allows us to resolve each comparison between $K^*$
and an arc $\delta \in D$, assuming that $\gamma \cap K^* \neq
\emptyset$, and we continue the binary search through each $M_j$ in
the same manner.

We again use a variant of Cole's technique~\cite{RC} to slightly
improve this bound further. In each parallel step we have a
collection $Z$ of $O(\log n)$ weighted points, each of which is an
intersection point of $\gamma$ with some arc $\delta$ from one of
the planar maps $M_j$, and we need to compare each of the points of
$Z$ with $K^*$. Let $D$ denote the set of these active arcs.

Note that each arc $\delta$ participating in this step
contributes (at most) four points to $Z$, for a total of at most
$4|D|$ points. We perform three steps of a (weighted) binary search on
the points of $Z$, where each step takes the weighted median $z_0$
of an appropriate portion of $Z$, and calls $\Pi_1(l_0)$, where $l_0$
is the vertical line through $z_0$. These $\Pi_1$-steps resolve the
comparisons with $K^*$ of all but $1/8$ of the points of $Z$, that is,
at most $(1/8)\cdot 4|D| = |D|/2$ points of $Z$ are still unresolved.

In other words, after the three calls
to the algorithm for solving $\Pi_1$ (in the first parallel
step of the execution), we can determine the outcomes of the
comparisons of at least half of the arcs in $D$ with $K^*$. We can
then proceed in this manner and apply Cole's technique (as before),
by using only a constant number of calls to $\Pi_1$ in each of the
$O(\log n)$ parallel steps of searching in all the maps.
This reduces a $\log \log n$ factor from the bound of the
running time, so it is only $O(\log^5n)$ time.

When these searches terminate, we end up with a 2-face in each
$M_j$, in which $\gamma \cap K^*$ lies (if nonempty), and we reach
the scenario described in a preceding paragraph. As explained there,
we can now either determine that $K \neq \emptyset$, or that $K =
\emptyset$, or else we know which side of $\gamma$, within $\psi_i$
(or, rather, within $\psi'$) can contain $K^*$, and we continue the
binary search through $\psi_i$ on that side.

When the binary search through $\psi_i$ terminates, we have a 2-face
$\zeta_i$ of $M_i$, where $K^*$ must lie, and we retrieve the ball
$b_i$ corresponding to $\zeta_i$. We repeat this step to each of the
maps $M_i^+$ and $M_i^-$ of each of the $t$ spherical polytopes
$S_i$, and obtain a set $\B_1$ of $2t$ balls. In addition, the
searches through the maps $M_i^+$ and $M_i^-$ may have trimmed
$\psi'$ to a narrower strip $\psi''$, and have produced a set $\B_2$
of witness balls, so that the $xy$-projection of their intersection
lies inside $\psi''$. $\B_2$ may consist of a total of $O(t^3 \log^2
n)$ witness balls, as is easy to verify. In addition, the
second-level searches produce an additional collection $\B_2'$,
consisting of balls corresponding to faces of the maps $M_j^+$ and
$M_j^-$, in which the second-level searches have ended; their
overall number is $O(t^2 \log n)$. Put $K_2 = \bigcap (\B_1 \cup
\B_2 \cup \B_2')$. Hence $K \neq \emptyset$ if and only if $K_2 \neq
\emptyset$.

As already noted, the overall running time of the emptiness
detection is $O(\log^5 n)$.

So far, we have only determined whether $K$ is empty or not.
However, to enable the decision procedure to discriminate between
the cases $r^* = r$ and $r^* < r$ we need to refine the algorithm,
so that it can also determine whether $K$ has nonempty interior (we
refer to an intersection $K$ with this property as
\emph{non-degenerate}).  To do so, we make the following
modifications to the algorithm described above. Each step in the
emptiness testing procedure which detects that $K \neq \emptyset$
obtains a specific point $w$ that belongs to $K$. Moreover, $w$
belongs to the intersection $K_1$ of polylogarithmically many
witness balls, and does not lie on the boundary of any other ball.
This is because each of the procedures $\Pi_0, \Pi_1$, or $\Pi_2$
locates the $xy$-projection $w^*$ of $w$ (which, for $\Pi_1$ and
$\Pi_2$ is a generic, unknown point in $K$) in each of the maps
$M_i^+, M_i^-, i = 1,\ldots,t$, and the collection of the witness
balls gathered during the various steps of the searches contains all
the balls that participate in the corresponding spherical polytopes
$S_i$ on whose boundary $w$ can lie. Thus, when we terminate with a
point $w \in K$, we find, among the polylogarithmically many witness
balls, the at most four balls whose boundaries contain $w$ (recall
our general position assumption), and test whether their
intersection is the singleton $\{w\}$. It is easily checked that
this is equivalent to the condition that $K$ is degenerate.

This completes the description of the algorithm, and concludes the
proof of Proposition~\ref{prop:intersection_test}.

\section{Discussion and Open Problems.}
\label{sec:discussion} In this paper we presented two algorithms for
computing the 2-center of a set of points in $\reals^3$. The first
algorithm takes near-cubic time, and the second one takes
near-quadratic time provided that the two centers are not too close
to each other. Note that our second algorithm may be slightly
revised, so that it receives, in addition to $P$, a parameter
$\epsilon > 0$ as input, and returns a solution for the 2-center
problem for $P$, if $\epsilon \leq 1-\frac{r^*}{r_0}$. To this end,
we run the exponential search until we reach a value of $r$ with
$1-\frac{r}{r_0} \leq \epsilon$. If along the search we have found a
value of $r$ such that $r \geq r^*$, we stop the search and run
Chan's technique with the constraint that $r^* \leq r$, as above.
Otherwise, we have $r^* > r_0(1-\epsilon)$ and we may return the
smallest enclosing ball of $P$ as an $\epsilon$-approximate solution
for the 2-center problem. This way, we ensure that the running time
of our algorithm is $O(\epsilon^{-3} n^2\log^5 n)$.

An obvious open problem is to design an algorithm for the 2-center
problem that runs in near-quadratic time on all point sets in
$\reals^3$.  Another interesting question is whether the 2-center
problem in $\mathbb{R}^3$ is \emph{{\sc 3sum}-hard} (see~\cite{GO}
for details), which would suggest that a near-quadratic algorithm is
(almost) the best possible for this problem.


\end{document}